\documentclass[conference,letterpaper]{IEEEtran}\IEEEoverridecommandlockouts

\usepackage{cite}
\usepackage[letterpaper, top=0.75in, bottom=1.0in, left=0.625in, right=0.625in, columnsep=0.25in]{geometry}
\usepackage{amsmath,amssymb,amsfonts}
\usepackage{algorithmic}
\usepackage{graphicx}
\usepackage{textcomp}
\usepackage{xcolor}
\def\BibTeX{{\rm B\kern-.05em{\sc i\kern-.025em b}\kern-.08em
    T\kern-.1667em\lower.7ex\hbox{E}\kern-.125emX}}

\usepackage{amssymb}
\usepackage{amsthm}
\usepackage{tikz}
\usepackage{tikz-cd}
\newtheorem{definition}{Definition}
\newtheorem{example}{Example}
\newtheorem{algorithm}{Algorithm}
\newtheorem{proposition}{Proposition}
\newtheorem{theorem}{Theorem}
\newtheorem{corollary}{Corollary}
\usepackage{nicefrac}

\begin{document}

\title{Optimizing Bit-Labeling of Voronoi Constellations\\
\thanks{This research was supported by the U.S. Department of
Commerce’s National Telecommunications and Information Administration
(NTIA) under the Public Wireless Supply Chain Innovation Fund Grant
Program (Award 24-60-IF2415: ASPEN - Advanced Signal Processing
Enhancement for Next-Generation Open Radio Units), administered by the
National Institute of Standards and Technology.}
}
\makeatletter
\newcommand{\linebreakand}{%
  \end{@IEEEauthorhalign}
  \hfill\mbox{}\par
  \mbox{}\hfill\begin{@IEEEauthorhalign}
}
\makeatother
\author{\IEEEauthorblockN{Carilyn Rumrill}
\IEEEauthorblockA{\textit{Rampart Communications} \\
Linthicum Heights, MD\\USA \\
\small{carrie@rampartcommunications.com}}
\and
\IEEEauthorblockN{David Muzzey}
\IEEEauthorblockA{\textit{Rampart Communications} \\
Linthicum Heights, MD\\USA \\
\small{dmuzzey@rampartcommunications.com}}
\and
\IEEEauthorblockN{Connor Davis}
\IEEEauthorblockA{\textit{Rampart Communications} \\
Linthicum Heights, MD\\USA \\
\small{connor@rampartcommunications.com}}
\linebreakand
\IEEEauthorblockN{Stephen Mackes}
\IEEEauthorblockA{\textit{Rampart Communications} \\
Linthicum Heights, MD\\USA\\
\small{stephen@rampartcommunications.com}}
\and
\IEEEauthorblockN{Daniel Chew}
\IEEEauthorblockA{\textit{Rampart Communications} \\
Linthicum Heights, MD\\USA \\
\small{dchew@rampartcommunications.com}}
}


\maketitle

\begin{abstract}
We optimize the bit-to-symbol map of the $D_4$ and $E_8$ root lattices in reference to bit error rate, by implementing two searches on a performance metric, Hamming Density. The first search is a traditional gradient descent and the second is an efficient search using twisted theta functions. We hold other sources of lattice gain constant by fixing the lattice constellation, and consider basis changes that permute the integer labelings of the lattice points. After searching the possible basis matrices for $D_4$ and $E_8$, we found up to 0.09 dB of gain in $D_4$ bit error rate curves, and up to 0.36 dB of gain in $E_8$ compared to the standard bases commonly used in literature at a BER of $10^{-4}$, which is on the scale of gain from Gray coding uncoded QAMs. 
\end{abstract}


\section{Introduction}

\IEEEPARstart{L}{attice} codes constructed with the root lattices give substantial gain in terms of bit error rate (BER) and lattice error rate (LER) over uncoded QAMs~\cite{mirani2020low, sadeghi2025design}; however, they are rarely used in practice. One reason for this is the inability to Gray code lattice constellations~\cite{li2021designing}. The density of the lattice and the ability to minimize the number of bit differences between nearest neighbors are opposing goals. A lattice point in a dense lattice has too many nearest neighbors to include only single bit differences between them, as in QAMs. Motivated by this, we describe a method to improve the bit-to-symbol map, giving BER gain over the standard lattice bases commonly used in literature~\cite{conway2013sphere}. We focus on the $D_4$ and $E_8$ root lattices, using $A_2$ only for examples. 

 We use the following self-similar lattice code modulation and demodulation technique described in~\cite{conway2003fast}. That is, given bits for encoding, they are first Gray mapped to integer vectors (see Gray mapping in~\cite{knuth2011art}) modulo $r$ of length $n$, where $n$ is the dimension of the lattice and $r$ is the defined modulus. The integer vectors are mapped onto the lattice with a basis matrix $B$. These points are then reshaped to lattice point representatives existing within the $r$-scaled Voronoi region repositioned by a shift vector $\vec{h}$. The lattice points are split into IQ pairs and normalized for transmission across the additive white Gaussian noise channel (AWGN). Upon reception, the IQ pairs are reconstructed to $n$-length vectors and the shift vector is added. A quantizer is applied to approximate the closest lattice points. The inverse basis and inverse Gray map are applied to approximate the received bits. Throughout, we use the Gray labeling on integers modulo r given by \begin{equation}
 \label{graymap}
     z\mapsto z \veebar (z\gg1).
 \end{equation} 

Following~\cite{conway2013sphere}, we use the bases given by the columns of~\eqref{eq: a2b},~\eqref{eq: d4b}, and~\eqref{eq: e8b}, which we refer to as `standard bases.' We generate a shift vector $\vec{h}$ from the standard basis with $r=4$ and the iterative optimization algorithm~\cite{conway2003fast}, given in~\eqref{eq: 1},~\eqref{eq: 2}, and~\eqref{eq: 3}.

In a survey of techniques~\cite{li2021designing}, the encoding and labeling techniques from Feng et al.~\cite{feng2013algebraic} and Kurkoski~\cite{kurkoski2018encoding} are compared in terms of their Gray penalty (essentially Hamming Density below). However, their reduction of bit errors depends on the lattice constellation construction. The studied use case is an integer lattice shaped by a more dense lattice, motivated by less complex encoding and decoding operations. In our framework, which uses the same dense lattice for both coding and shaping, the suggested basis changes of these methods are trivial. In Feng's et al. method, the scaled identity is already in Smith Normal form, and in Kurkoski's method, the matrix $L$ is also a scaled identity matrix (see section 3 of~\cite{li2021designing}). Because these suggested alternatives reduce to the original encoding and indexing in~\cite{conway2003fast}, they will not improve the Gray penalty of the bit labeling. As a result, we suggest an entirely different approach to bit-labeling for a self-similar lattice code.

\begin{align}
A2:\quad 
B_0 &=
{
\setlength{\arraycolsep}{3pt}
\renewcommand{\arraystretch}{0.85}
\begin{pmatrix}
1 & -\nicefrac{1}{2} \\
0 & \nicefrac{\sqrt{3}}{2}
\end{pmatrix}
}
\label{eq: a2b}
\\[0.5em]
\vec{h} &= (\frac{1}{8}, \frac{\sqrt{3}}{8})\label{eq: 1}
\\[0.5em]
D4:\quad 
B_0 &=
{
\setlength{\arraycolsep}{3pt}
\begin{pmatrix}
1 & 1 & 0 & 0 \\
1 & -1 & 1 & 0 \\
0 & 0 & -1 & 1 \\
0 & 0 & 0 & -1
\end{pmatrix}
}
\label{eq: d4b}
\\[0.5em] \vec{h} &= (\frac{17}{32}, \frac{6}{32}, 0, \frac{-11}{32})\label{eq: 2}
\\[0.5em]
E8:\quad 
B_0 &=
{
\setlength{\arraycolsep}{2.5pt}
\begin{pmatrix}
2 & -1 & 0 & 0 & 0 & 0 & 0 & \nicefrac{1}{2} \\
0 & 1 & -1 & 0 & 0 & 0 & 0 & \nicefrac{1}{2} \\
0 & 0 & 1 & -1 & 0 & 0 & 0 &\nicefrac{1}{2} \\
0 & 0 & 0 & 1 & -1 & 0 & 0 & \nicefrac{1}{2} \\
0 & 0 & 0 & 0 & 1 & -1 & 0 &\nicefrac{1}{2} \\
0 & 0 & 0 & 0 & 0 & 1 & -1 & \nicefrac{1}{2} \\
0 & 0 & 0 & 0 & 0 & 0 & 1 & \nicefrac{1}{2} \\
0 & 0 & 0 & 0 & 0 & 0 & 0 & \nicefrac{1}{2}
\end{pmatrix}
}
\label{eq: e8b}
\\[0.5em]
\vec{h} &= (0.645, 0.0484,0.116, 0.214,\label{eq: 3}\\ &\hspace{.65cm}0.182, 0.247, 0.317, 0.083).\nonumber
\end{align}

 We report our results as follows: Section~\ref{sec:considering} introduces the space of marked lattice we consider. Section~\ref{sec:HD} defines Hamming Density (HD) with an example on the A2 lattice. Section~\ref{sec: hamming_decent} describes a gradient descent method used to search through the space of marked lattices to find minimal HD values. Section~\ref{sec: efficient_computation} introduces a twisted-theta series exhaustive search that can be used for examples with small to moderate parameter sizes. Section \ref{sec:lifting} provides a lifting criterion for shears of increasing moduli. Section~\ref{sec:berves} presents simulated results. Finally, Section~\ref{sec:conclusion} concludes our findings.



\section{Considering Marked Lattices}
\label{sec:considering}
To simplify discussion, we distinguish between a lattice and its labeling. 

\begin{definition}
    A lattice $\Lambda$ is the set of integral linear combinations of $n$ linearly independent vectors $b_1, b_2, ... b_n$. Arranging these basis vectors as the columns of a $n \times n$ matrix $B$, we write 
    \begin{equation}
    \Lambda = \{Bz | z \in \mathbb{Z}^n \}.
\end{equation}
\end{definition}

\begin{definition}
    A marked lattice is a lattice paired with a specific choice of basis, denoted by $(\Lambda, B).$
\end{definition}

\noindent We hold other sources of gain constant, beyond the bit-to-symbol map, in our comparison to the standard bases by considering multiple marked lattices for a given lattice. That is, we consider various labelings for a fixed constellation. 




We multiply unimodular matrices by fixed bases $B_0$ for each lattice to find distinct marked lattices, while remaining on the same unmarked lattice. A search through the group $\text{GL}_n (\mathbb{Z})$ will provide all distinct labelings for a given lattice. However, this is a highly redundant search since $\text{GL}_n (\mathbb{Z})$ is an infinite group, while our list of distinct lattice labelings is finite. This necessitates an appropriate equivalence relation that reduces marked lattices by their equivalent lattice labelings. We use the group isomorphism,



\begin{align}
     \text{GL}_n(\mathbb{Z})/\Gamma_r\label{eq: eqiv} \cong \text{GL}_n(\mathbb{Z}/r\mathbb{Z}) \\
    \text{where } \Gamma_r = (I+r\mathbb{Z}^{n\text{x}n})\cap\text{GL}_n\mathbb{Z}.
\end{align}

\noindent Equation \eqref{eq: eqiv} can be read as, two marked lattices produce the same labeling if their unimodulars reduce to the same matrix modulo $r$. As a shorthand, we will denote the left side of \eqref{eq: eqiv} as $G_n(r)$.

 We choose a generating set for this group consisting of only elementary shear matrices.  For reference, $G_4(r)$ can be presented using 12 shear generators  and $G_8(r)$ can be presented using 56.

The cardinality of $G_n(r)$ is expressed as, 
\begin{equation}
    |G_n(r)| = r^{n^2-1} \prod_i (1-\frac{1}{p_i^2})(1-\frac{1}{p_i^3})...(1-\frac{1}{p_i^n})
\end{equation}
where $p_i$ are unique prime factors of $r$ \cite{diamond2005first}.  
Note, the rate the cardinality of $G_n(r)$ grows as we increase $n$ and $r$ in Table~\ref{tab:group_orders}. For the parameters of interest, $G_n(r)$ is too large to directly exhaust through, thus we introduce a sub-exhaustive search in Section~\ref{sec: hamming_decent}, and conduct a global search by exploiting a related cost function on a much smaller search space  for $r=4$ in Section~\ref{sec: efficient_computation}, then utilize a lifting criterion from $r=4$ to $r>4$ in Section~\ref{sec:lifting}. 




\begin{table}[htbp]
\caption{$|G_n(r)|$ for given values of $n$ and $r$}
\begin{center}
\begin{tabular}{|c|c|c|c|}
\hline
\textbf{} & \textbf{$n = 2$} & \textbf{$n = 4$} & \textbf{$n = 8$} \\ 
\hline
$r = 2$ & $6$ & $2.016 \times 10^4$ & $5.3481 \times 10^{18}$ \\ 
\hline
$r = 4$ & $48$ & $6.606 \times 10^8$ & $4.9327 \times 10^{37}$ \\ 
\hline
$r = 8$ & $384$ & $2.1647 \times 10^{13}$ & $4.5496 \times 10^{56}$ \\ 
\hline
$r = 16$ & $3,072$ & $7.0932 \times 10^{17}$ & $4.1963 \times 10^{75}$ \\ 
\hline
\end{tabular}
\label{tab:group_orders}
\end{center}
\end{table}

\section{Hamming Density}
\label{sec:HD}
Rather than running a BER curve or a single BER point for every considered unimodular, we introduce a faster performance metric for a marked lattice (see Gray penalty \cite{li2021designing}).

\begin{definition}
    Given a marked lattice with a self-similar code construction, we call the total number of bit differences of a lattice point to all their nearest neighbors (including neighbors wrapped around the scaled Voronoi region) divided by the number of nearest neighbors the Hamming Density (HD).
\end{definition}

\noindent Proposition \ref{prop} below claims the HD is well defined, meaning HD is constant across every point in a marked lattice, as long as a certain criterion is met. We discuss this criterion in Definition \ref{def:in_val} and Proposition \ref{prop}. For now, this criterion is met for all $D_4$ bases considered, as well as for all $E_8$ bases considered with $r\le 8$.
Thus, we can define the HD associated with marked lattice by the bit differences of a single lattice point with its neighbors divided by the number of nearest neighbors. For simplicity, we chose the lattice point generated from the zero vector. As such, the HD approximates the bit errors per lattice error by ignoring lattice errors outside the nearest neighbor set. For reference,  a lattice point on $A_2$ has 6 nearest neighbors, $D_4$ has 24 nearest neighbors, and $E_8$ has 240 nearest neighbors.

\begin{example}

For simplicity, we will refer to lattice points by their integer labelings in the following example. Consider two $A_2$ $r=4$ marked lattices shown in Figure \ref{fig:shear_example}. The bottom labels are made with standard basis and top labels are made with the the product of the standard basis and the unimodular $\begin{pmatrix}
    2 & 3 \\
    1 & 2
\end{pmatrix}$. Consider the nearest neighbors of the zero vector. (Again, this is the lattice point with the integer labeling $\vec{0}$). The standard basis has nearest neighbors: 

\begin{equation*}
(01),(11),(10),(03),(33), \text{and } (30).
\end{equation*}

\noindent If we convert these integers to bits and reverse the Gray map we get: 

\begin{equation*}
(0001),(0101),(0100),(0010),(1010),\text{and }(1000).
\end{equation*}

\noindent The zero integer vector has 8 bit differences between it and it's integer neighbors. The HD tells us there are approximately $\frac{8}{6} = 1.33$ bit differences per nearest neighbor in this marked lattice. 

\noindent If we do the same for the sheared basis we get 16 bit differences and HD $\frac{16}{6} = 2.67.$

\end{example} 

Here, we introduce a technical quantity that governs when adding Gray-mapped integer coordinates of short vectors will have the same effect, regardless of the other summand. 
\begin{definition}
\label{def:in_val}
    Given a marked lattice $(\Lambda,B)$ where $B=B_0U$ for some $U\in GL_n(r),$ we define the integer width of $U$ with respect to $B_0$ by \begin{equation}
        \omega(U)=\max_{0\leq j\leq n,\alpha\in\Phi}|U_j^{-1}\cdot z(\alpha)|
    \end{equation}
    where $\Phi$ is the set of short vectors in $\Lambda$, $U_j^{-1}$ are the rows of $U^{-1}$, and $z(\alpha)=B_0^{-1}\alpha$ are the integer coordinates of $\alpha$ with respect to $B_0$. That is, the integer width is the largest absolute value of an integer coordinate of any short vector in the new coordinates associated with $U$.
\end{definition}

\begin{proposition}
\label{prop}
Let $(\Lambda,B)$ be a marked lattice and $B=B_0U$ where $\omega(U)\leq4.$
Then the number of bit differences is constant when averaged between the nearest neighbors of any constellation point.  
\end{proposition}

\begin{proof}
Consider the lattice point generated from the zero vector $\vec{l}_0 = z^{-1}(\vec{0})$ in $(\Lambda,B)$. Let $\Phi$ denote the set of nearest neighbors of $\vec{l_0}$ and let $\mathcal{A}$ denote the integer vectors of these neighbors
\begin{equation} \mathcal{A} = z(\Phi).\end{equation} 
We know for a given lattice point $\vec{l}_k$, the nearest neighbors $\Phi_k$ of that lattice point can be found by 
\begin{align}
    \Phi_k=\{z^{-1}(z(\vec{l}_k)+\vec{a}~\text{mod}~r)|\vec{a}\in\mathcal{A}\}\label{nonzeronn},
\end{align}
since we have assumed the constellation is shaped around a scaled Voronoi region. 
Equation \ref{nonzeronn} affirms that every lattice point has the same integer differences $\mathcal{A}$ to its nearest neighbors, since $\mathcal{A}$ does not depend on $k$.

Note that the maximum of the magnitude of these integer differences was defined as the integer width of the marked lattice. Thus, integer width is also constant across the marked lattice.

It can be shown that for the Gray code defined in  \eqref{graymap}, all integer differences $\le4$ are associated with the same number of average bit differences (averaged over the positive and negative directions) regardless of starting integer. Therefore, bases that have integer width $\le4$ have a uniform HD across the entire lattice. This is trivially true for all $r\le8$, and true for the optimal bases presented in this paper for $r\ge16$.
\end{proof}

\begin{figure}[ht]
    \vspace{1mm}
    \centering
    \includegraphics[width=.75\linewidth]{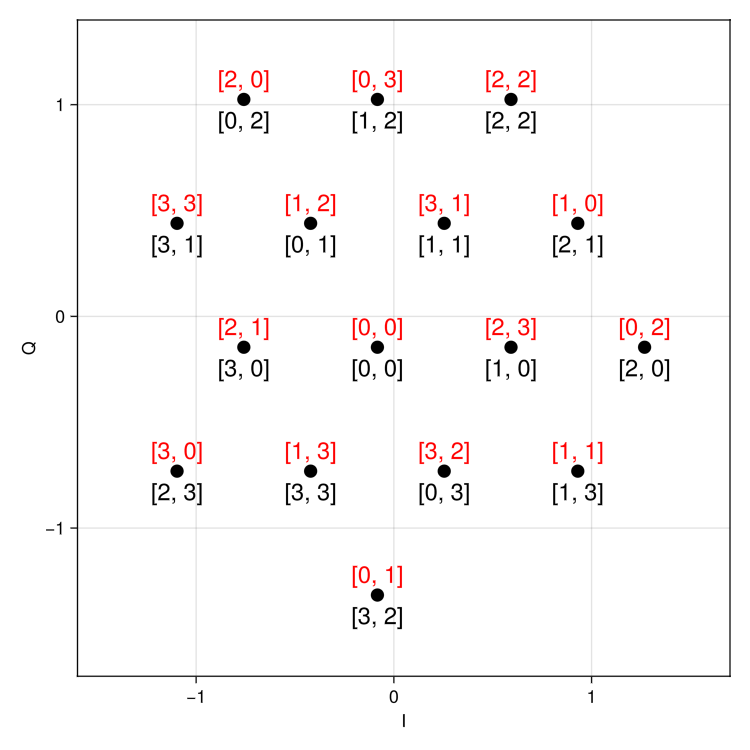}
    \caption{$A_2$ $r=4$ marked lattice generated by the standard basis (bottom) and the standard basis sheared by $\begin{pmatrix}
    2 & 3 \\
    1 & 2
\end{pmatrix}$ (top). }
    \label{fig:shear_example}
\end{figure}

Note, $\mathcal{A}$ depends on the basis matrix. So, $\mathcal{A}$ must be recomputed with every unique element in $G_n(r)$, but need only be computed with one point. Given a lattice point's integer vector using the $A_2$ standard basis, the integer vector of the nearest neighbor can be found by adding one of the following modulo $r$:
 $\mathcal{A}_s = (1,0)$, $(1,1)$, $(0,1)$, $(-1,0)$, $(-1,-1)$, $(0,-1)$. See this is true in Figure \ref{fig:shear_example}. Reduction modulo $r$ ensures this is true even after wrapping around the lattice constellation. 

A sheared basis has a different set of integers. Specifically, 
\begin{equation}
    \mathcal{A}_u = U^{-1}\mathcal{A}_s
\end{equation} where $U$ is the considered unimodular. We can calculate the HD of a given marked lattice by the following:\\
\begin{algorithm}
\label{alg:hd}
\leavevmode
   \begin{enumerate}
    \item Choose a unimodular U.
    \item Calculate $\mathcal{A}_u$ mod $r$.
    \item Expand integer vectors in $\mathcal{A}_u$ to bit vectors.
    \item Undo Gray mapping.
    \item Calculate the HD between the bit vectors in 4) and $\vec{0}$.
\end{enumerate} 
\end{algorithm}

\section{Gradient Descent valued by Hamming Density}
\label{sec: hamming_decent}

We present here a heuristic search algorithm that is broadly applicable to any lattice, and focus on the results from the two lattices investigated in this paper $D_4$ and $E_8$. We search for a minimum HD with a gradient descent that travels in the direction of elementary shears ($G_n(r)$ generators) and is evaluated by HD. 
Consider the following steps:
\begin{algorithm}
\leavevmode
\begin{enumerate}
    \item Choose a unimodular and call it $G_m$.
    \item Take all the positive generators $G_i$ and compute:
    $G_mG_i, G_mG_i^2, \dots , G_mG_i^{r-1} $.
    \item Calculate the HDs of the products.
    \item Set the product with the first occurrence of the smallest HD as the new $G_m$.
    \item Do 1) - 4) a number of times (we chose 25) and store the smallest HD at each step.
    \item Do 1) - 5) for all a selection of initial unimodulars. We choose a collection of unimodulars with unique HD. 
\end{enumerate} 
\end{algorithm}


    




            

We conducted this gradient descent on $D_4$ $r = 2\text{ and }4$ and $E_8$ $r=2 \text{ and }4$. In Section \ref{sec:lifting}, we show that the $r=4$ shears can be lifted for $r=8 \text{ and } 16$ results. We initialized this search with unique HD values produced from the product of the standard basis with elementary shears. $D_4$ $r=2$ and $4$ have 2 and 5 unique HD values respectively. $E_8$ $r=2$ and $4$ have 3 and 29 unique HD values respectively. We also included the standard basis as an initial value.

The smallest HD values seen in this descent and the $r=8 \text{ and } r=16$ lifts are compared to the standard basis HD in Table \ref{tab:hd_table}. In Section \ref{sec: efficient_computation}, we will show these are the optimal HD values for these lattices. However, with the descent alone, we recover many local minima, and had no way to know if our smallest seen HD value was optimal. This motivated a search for  a more efficient HD computation, found in Section \ref{sec: efficient_computation}.




\begin{table}[htbp]
\caption{HD values for the standard basis and optimally sheared basis}
\begin{center}
\begin{tabular}{|c|c|c|c|c|}
\hline
\textbf{} & \multicolumn{2}{c|}{\textbf{$D_4$ Lattice}} & \multicolumn{2}{c|}{\textbf{$E_8$ Lattice}} \\
\cline{2-5}
\textbf{\textit{r}} & \textbf{\textit{Standard HD}} & \textbf{\textit{Minimum HD}} & \textbf{\textit{Standard HD}} & \textbf{\textit{Minimum HD}} \\ 
\hline
$2$  & $2.1667$ & $2$ & $3.9333$ & $3.7333$ \\ 
\hline
$4$  & $2.3333$ & $2$ & $7.3500$ & $3.8667$ \\ 
\hline
$8$  & $2.3333$ & $2$ & $9.1333$ & $3.8667$ \\ 
\hline
$16$ & $2.3333$ & $2$ & $9.4333^{\mathrm{*}}$ & $3.8667$ \\ 
\hline
\multicolumn{5}{l}{$^{\mathrm{*}}$HD approximation since integer width bound is not satisfied.}
\end{tabular}
\label{tab:hd_table}
\end{center}
\end{table}

\subsection{Resulting Matrices}
We include the unimodulars with the lowest found HD for each $r$ considered in~\eqref{eq:d4shear} and~\eqref{eq:e8shear}.

\begin{align}
& D_4\hspace{1mm} r = 2,4,8,16 \text{ } HD = 2 \nonumber\\[0.4em]
& U =
{
\setlength{\arraycolsep}{3pt}
\begin{pmatrix}
 1 & 0 & 0 & 0 \\
 0 & 1 & 0 & 0 \\
-1 & 0 & 1 & 0 \\
 0 & 0 & 0 & 1
\end{pmatrix}
}
\label{eq:d4shear}
\\[0.6em]
& E_8\hspace{1mm} r = 2,4,8,16 \text{ } HD = 3.86 \nonumber\\[0.4em]
& U =
{
\setlength{\arraycolsep}{3pt}
\begin{pmatrix}
-1 &  1 & 0 & 0 & 0 & 1 & 1 & 0 \\
-2 &  1 & 0 & 0 & 0 & 2 & 2 & 0 \\
-2 &  0 & 1 & 0 & 0 & 2 & 2 & 0 \\
-1 &  0 & 0 & 1 & 1 & 1 & 1 & 0 \\
-1 &  0 & 0 & 0 & 1 & 1 & 1 & 0 \\
 0 &  0 & 0 & 0 & 0 & 1 & 1 & 0 \\
 0 &  0 & 0 & 0 & 0 & 0 & 1 & 0 \\
 0 & -1 & 0 & 0 & 0 & 0 & 0 & 1
\end{pmatrix}
}
\label{eq:e8shear}
\end{align}

\section{Efficient Computation of Hamming Density}
\label{sec: efficient_computation}
To efficiently conduct a search, one must compute the HD of many candidates. To that end we derive a method for efficient computation of HD as well as providing concise closed formulae for certain lattices. Such closed formulae may not be accessible for every lattice of general interest. In such cases, the gradient descent provides a reliable fallback. We begin by giving an expression for the HD.
\begin{definition}
    For an integer vector $v,$ define $c(v)$ to be the number of 1s in the binary expansions of the Gray coded integer digits. 
\end{definition}
\begin{definition}
    For a unimodular matrix, U and corresponding basis $B=B_0U$, define $C(U),$ the cost of $U$ by 
    \begin{equation}
    C(U)=\frac{1}{|\Phi|}\sum_{\alpha\in\Phi} c(U^{-1}z(\alpha))
    \end{equation}  
\end{definition}
This follows immediately from Algorithm~\ref{alg:hd} and Proposition~\ref{prop}.

By definition, we have \begin{equation}
    C(U)=\text{HD}(B_0U).
\end{equation}
From the construction in terms of $c(v)$ we can see that the components of each vector that contribute to the overall vector cost are related to the rows of $U^{-1}.$  For full rank lattices, this is the action of the unimodular matrix $U$ to create a new dual basis, $(BU)^{-T}$. So the cost of a row of $U^{-1}$ is determined by its dot product with the integer coordinates of the short vectors.

Fix $r=4$. For $u\in\mathbb{Z}_4^n,$ define \begin{equation}
    \hat{R}(u)=\sum_{z\in B_0^{-1}\Phi}e^{\frac{\pi}{2}i u\cdot z}=\sum_{\alpha\in\Phi}\cos\left(\frac{\pi}{2}u\cdot z(\alpha)\right),
\end{equation}
where we note that the dot product is taken $\mod r,$ and the result is well-defined due to the periodicity of cosine. The interested reader can observe that this is the $\mathbb{Z}_4$ Fourier transform of the indicator function on the set of short vector integer coordinates. That is $1_{\{z(\alpha)|\alpha\in\Phi\}}$. For more details see (e.g.)~\cite{terras1999fourier}.  Letting $N_j(u)=\#\{\alpha\in\Phi|u\cdot z(\alpha)=j\mod4\},$  and noting that integers 1, and 3 contribute one bit difference to the cost, integer 0 contributes zero, and integer 2 contributes three bit differences, we can use the symmetry of the lattice and the fact that $3=-1\mod4$ to yield \begin{equation}
    \hat{R}(u)=|\Phi|-c(u)=N_0-N_2,
\end{equation}  
which follows from $|\phi|=\sum_jN_j.$ If $u_j$ are the rows of $U^{-1}$ then we recover the relationship between $\hat{R}$ and HD. That is 
\begin{equation}
    \text{HD}(B)=n-\frac{1}{|\Phi|}\sum_{j=1}^n\hat{R}(u_j).
\end{equation}
\begin{definition}
    We define the twisted theta function $\Theta(\ell,\tau)=\sum_{x\in\Lambda}e^{i\frac{\pi}{2}\ell\cdot x+2\pi i\tau\lVert x\rVert^2}$ where $\ell=B_0^{-T}u\in\Lambda'$ for some candidate row $u$ of $U^{-1}$.  
\end{definition}
By making the substitution $q=e^{2\pi i\tau}$ as is typical of lattice theta functions~\cite{conway2013sphere}, we can see that $\hat{R}(u)$ is the coefficient in the expansion of $\Theta$ corresponding to the power of q that is half the length of a short vector. We can use this relationship to recover closed formulas for $\hat{R}$ on certain lattices. Here, we supply such formulas for $D_4,E_8.$ In both cases we proceed by making the substitution $v=u/4$,$\tau'=\tau/2,$ and then writing the twisted theta functions in terms of Jacobi theta functions with well-known expansions~\cite{conway2013sphere}. This strategy yields 
\begin{align}
    \hat{R}_{D_n}(u)&=2\left(\left(\sum_{j=1}^n\cos(\frac{\pi}{2} u)\right)^2-\sum_{j=1}^n\cos^2(\frac{\pi}{2} u)\right) \label{dntheta}\\
    \hat{R}_{E_8}(u)&=\hat{R}_{D_{8}}(u)+128\left(\prod_{j=1}^8\cos(\frac{\pi}{4}u )+\prod_{j=1}^8\sin(\frac{\pi}{4}u)\right) \label{e8theta},
\end{align}
but may not be applicable for every lattice of general interest. 
Further, since this operation is performed on the rows of $U^{-1}$ we massively reduce the search space from unimodular matrices to the space of possible rows $\mathbb{Z}_r^n,$ with the constraint that we must select a full rank (independent modulo 2) set of rows. 

Using this strategy, we can fully exhaust the space of rows and in doing so we recover the same HDs as the gradient descent search (Section \ref{sec: hamming_decent}), indicating that the original search was able to find a globally optimal shears for the $D_4$ and $E_8$ lattices when $r=4$. 
\section{A Lifting Criteria for Higher Modulation Orders}
\label{sec:lifting}
We now show that if a basis is found with a particular HD so that the rows of the corresponding shear stay inside a certain polytope generated by the short vectors of the lattice, then that HD is constant for $r\geq 4.$
\begin{theorem}\label{diffR}
For a marked lattice $(\Lambda, B=B_0U)$ with $U\in GL(n,\mathbb{Z}),$ If $\omega(U)\leq2$ for all $\alpha\in\Phi,$ then $C(U)$ is constant across $r\geq4.$ 
\end{theorem}
\begin{proof}Suppose the dot product of the $GL(n,Z)$ lift with all short vectors has $|U^{-1}_j\cdot z(\alpha)|\leq 2.$ Moving from $0$ to $\pm1,$ or $\pm1$ to $\pm2$, one only has the option to add a bit. Thus, within the assumed constraint, the cost contribution of this specific dot product reduced modulo $r$ is simply equal to its absolute value, and this is true for all $r\geq 4.$ Since this is true for each row and each root, by construction the HD is also preserved across $r\geq 4.$
\end{proof}
\begin{corollary}
    For $r\geq4$ the optimal HD among bases with $\omega(U)\leq4$ of $D_4$ is 2, and for $E_8,$ it is $3.866...$.
\end{corollary}
\begin{proof}
    If there is a choice of shear for $r=2^{k+1}$ with HD=$\rho$ we can recover one for $r=2^k$ by restricting the domain to points with integer coordinates $z\leq2^{k}$ and by removing a bit from the labeling. This preserves $HD=\rho$ as long as the number of nearest neighbours is preserved ($r\geq4$). Thus by induction for $r\geq 4$ the HD is weakly increasing. By Theorem~\ref{diffR} it is weakly decreasing. Thus it is constant for $r\geq4.$ Global search using Equations~\ref{dntheta} and~\ref{e8theta} yielded shears that satisfied the constraint of Theorem~\ref{diffR}.
\end{proof}

\section{Performance in a White Noise Channel}
\label{sec:berves}

To benchmark the effect of bit-labeling, we constructed QAMs of various modulation orders with non-Gray bit-to-symbol maps and compared their uncoded performances in AWGN to QAMs with traditional Gray codes. These non-Gray bit-to-symbol maps were constructed by assigning half of the bits to symbols in sequential order along each I and Q dimension independently. Bit error rate curves for these QAMs are shown in Figure \ref{fig:z1berves}. $E_b/N_0$ gain is reported in Figure \ref{fig:gainxmodulus} via interpolating the $E_b/N_0$ value that yields a BER = $1 \times 10^{-4}$, and subtracting the Gray coded $E_b/N_0$ from the Sequential coded $E_b/N_0$.

For lattice modulations, we took the unimodulars with the optimal HD and compared their bit error rate curves with that of their corresponding standard bases. For $D_4$, in Figures~\ref{fig:d4berves} and~\ref{fig:gainxmodulus} we see the HD = 2 unimodular has modest gain above the standard HD = 2.33 basis, between 0.05 dB  and 0.09 dB. While gain does improve as lattice modulation order increases, even for the 256-QAM equivalent $D_4$ $r=16$ case, the observed gain is less than 0.1 dB.

For $E_8$, we see in Figures~\ref{fig:e8berves} and~\ref{fig:gainxmodulus} the unimodular with HD = 3.86 gives about 0.25 dB of gain at $10^{-4}$ over the standard HD = 7.35 unimodular with $r=4$. The HD of the standard basis increases in $r=8$ to HD = 9.13 and again increases for $r=16$ to HD $\approx$ 9.43, giving an $E_b/N0$ gain around 0.33-0.36 dB respectively. As shown in Figure~\ref{fig:gainxmodulus}, simply changing the $E8$ basis yields similar increases in gain compared to equivalent QAM Gray vs sequential coding. 


\begin{figure}[ht]
    \centering
    \includegraphics[width=1\linewidth]{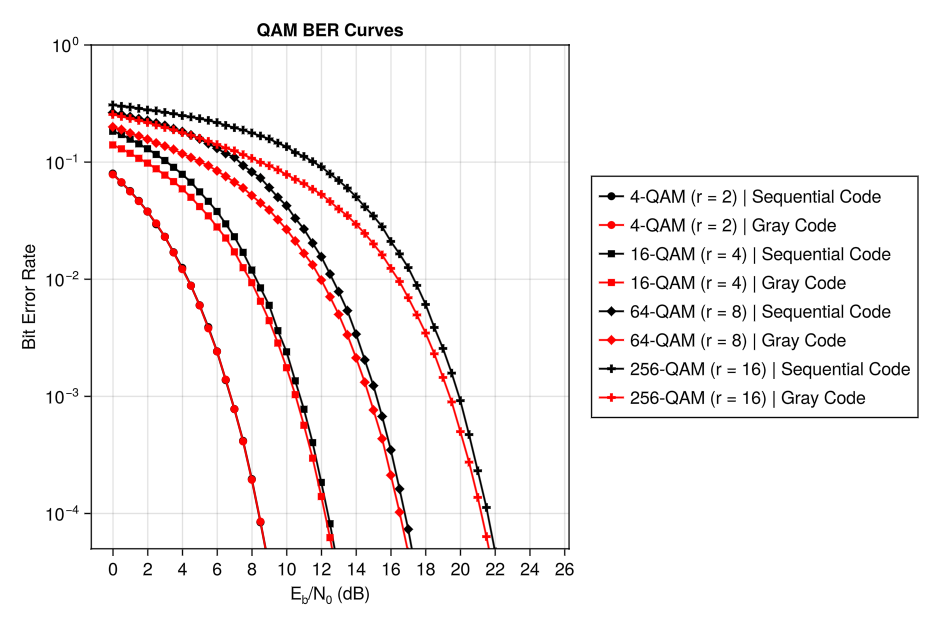}
    \caption{QAM bit error rate curves across r values, comparing the standard Gray coding to a non-Gray Sequential code.}
    \label{fig:z1berves}
\end{figure}

\begin{figure}[ht]
    \centering
    \includegraphics[width=1\linewidth]{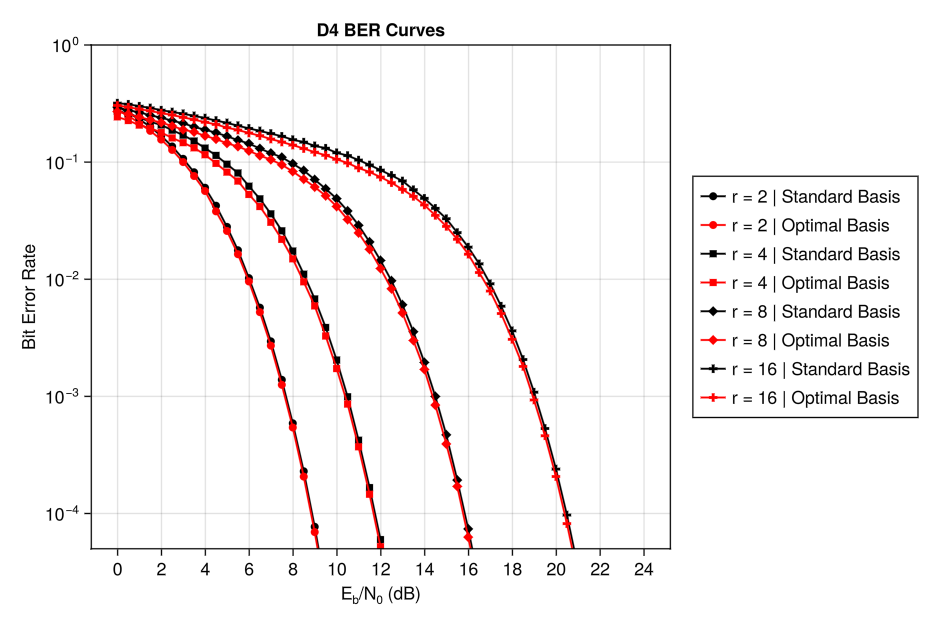}
    \caption{$D_4$ bit error rate curves across r values, comparing the standard basis to the best unimodular.}
    \label{fig:d4berves}
\end{figure}

\begin{figure}[ht]
    \centering
    \includegraphics[width=1\linewidth]{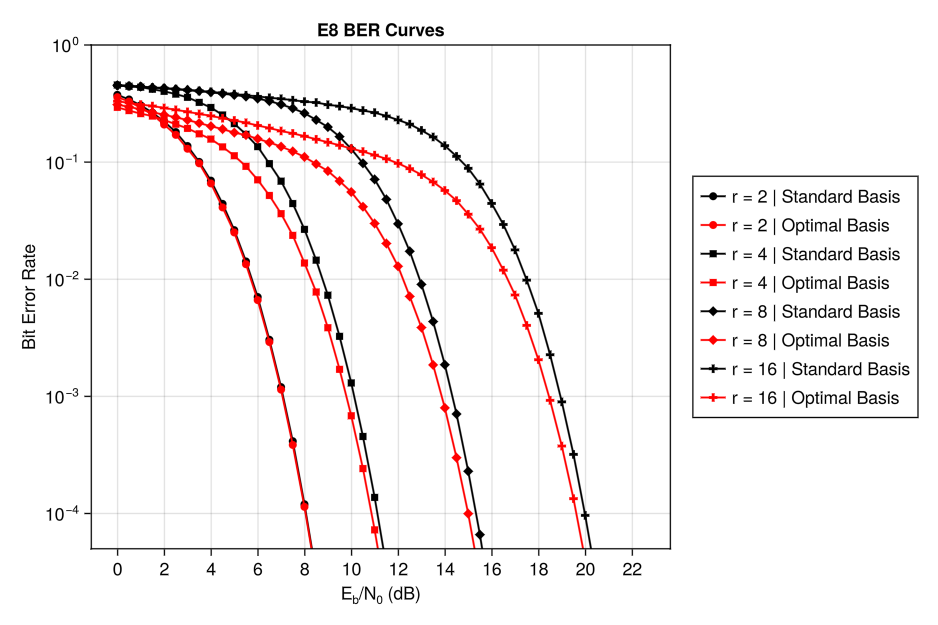}
    \caption{$E_8$ bit error rate curves across r values, comparing the standard basis to the best unimodular.}
    \label{fig:e8berves}
\end{figure}

\begin{figure}[ht]
    \centering
    \includegraphics[width=1\linewidth]{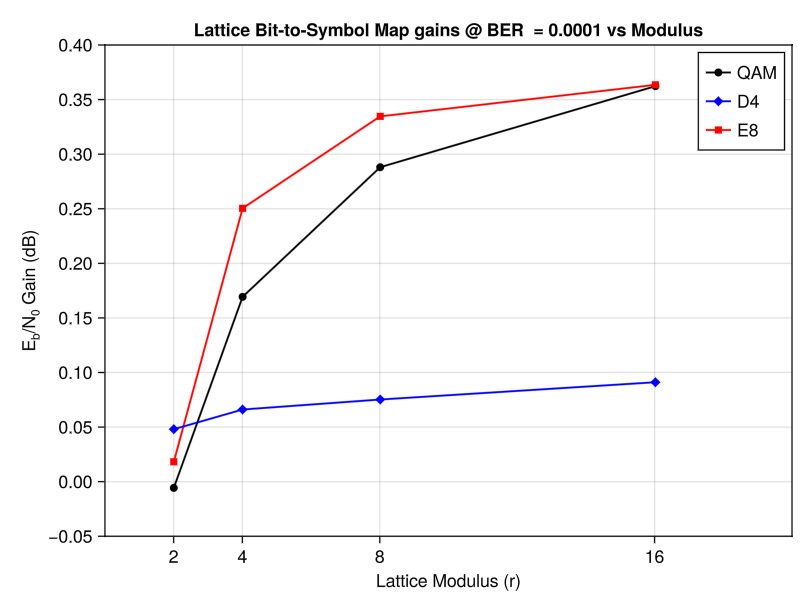}
    \caption{Interpolated $E_b/N_0$ gain (dB) for bit-to-symbol map improvements for QAM, $D_4$, $E_8$ as modulation order increases.}
    \label{fig:gainxmodulus}
\end{figure}

\section{Conclusion}
\label{sec:conclusion}
Using a gradient descent search, we have found BER gain over the standard D4 and E8 bases found in literature. For E8, this gain is on the scale of gain seen from Gray coding uncoded QAMs. Although the gradient descent witnessed a very small sample of the total space of lattice labelings, the search was able to find the global HD minima proven by twisted theta functions. This encourages a similar approach for 
lattices of higher dimension and modulus, where an exhaustive search cannot be derived.

Avenues for future work include measuring the effect of improved bit-labeling on soft information estimation, an alternative bit-labeling strategy for set-partition coded modulations. The relationship between lattice geometry and realized BER gain motivates an exploration of the best bases for other dense lattices. Finally, mature numerical techniques could improve the speed of the descent search and a study of the dynamical properties of the search could discuss convergence.

\bibliographystyle{IEEEtran}
\bibliography{bibliography}

\end{document}